\documentclass[conference]{IEEEtran}
\usepackage{cite}
\usepackage{amsmath}
\usepackage{multirow}
\usepackage{epsfig}
\usepackage{mathrsfs}
\usepackage{amssymb}
\usepackage{comment}
\usepackage{array}
\usepackage{blkarray}
\usepackage{enumerate}
\usepackage{url}
\usepackage{chemarrow}
\usepackage{color}
\usepackage[tight,footnotesize]{subfigure}
\usepackage{fancyhdr}

\makeatletter
\def\ps@headings{%
\def\@oddhead{\mbox{}\scriptsize\rightmark \hfil \thepage}%
\def\@evenhead{\scriptsize\thepage \hfil \leftmark\mbox{}}%
\def\@oddfoot{}%
\def\@evenfoot{}}
\makeatother
\renewenvironment{thebibliography}[1]{%
   \begin{oldthebibliography}{#1}%
     \setlength{\itemsep}{-.3ex}%
}%
{%
   \end{oldthebibliography}%
}

\makeatletter
\newif\if@restonecol
\makeatother

\usepackage[lined,ruled,linesnumbered]{algorithm2e}
\makeatletter
\def\old@comma{,}
\catcode`\,=13
\def,{%
  \ifmmode%
    \old@comma\discretionary{}{}{}%
  \else%
    \old@comma%
  \fi%
}
\makeatother

\newtheorem{theorem}{Theorem}

\newtheorem{corollary}[theorem]{Corollary}
\newtheorem{claim}[theorem]{Claim}

\newtheorem{definition}[theorem]{Definition}

\DeclareMathAlphabet{\mathcal}{OMS}{cmsy}{m}{n}

\usepackage[tight,footnotesize]{subfigure}

\begin{document}
\pagestyle{fancy}
\IEEEoverridecommandlockouts

%
%

\lhead{\textit{Technical Report, IBM T. J. Watson Research Center, Yorktown, NY, USA, October, 2017.}}
\rhead{} 

\title{Influence Maximization Under Generic Threshold-based Non-submodular Model}
\author{\IEEEauthorblockN{Liang Ma}
\IEEEauthorblockA{IBM T. J. Watson Research\\
Yorktown, NY, USA\\
Email: maliang@us.ibm.com}
}

%
%

\maketitle
\IEEEpeerreviewmaketitle

\begin{abstract}
As a widely observable social effect, influence diffusion refers to a process where innovations, trends, awareness, etc. spread across the network via the social impact among individuals. Motivated by such social effect, the concept of influence maximization is coined, where the goal is to select a bounded number of the most influential nodes (seed nodes) from a social network so that they can jointly trigger the maximal influence diffusion. A rich body of research in this area is performed under statistical diffusion models with provable submodularity, which essentially simplifies the problem as the optimal result can be approximated by the simple greedy search. When the diffusion models are non-submodular, however, the research community mostly focuses on how to bound/approximate them by tractable submodular functions so as to estimate the optimal result. In other words, there is still a lack of efficient methods that can directly resolve non-submodular influence maximization problems. In this regard, we fill the gap by proposing seed selection strategies using network graphical properties in a generalized threshold-based model, called influence barricade model, which is non-submodular. Specifically, under this model, we first establish theories to reveal graphical conditions that ensure the network generated by node removals has the same optimal seed set as that in the original network. We then exploit these theoretical conditions to develop efficient algorithms by strategically removing less-important nodes and selecting seeds only in the remaining network. To the best of our knowledge, this is the first graph-based approach that directly tackles non-submodular influence maximization. Evaluations on both synthetic and real-world Facebook/Twitter datasets confirm the superior efficiency of the proposed algorithms, which are orders of magnitude faster than benchmarks for large networks.
\end{abstract}
\begin{IEEEkeywords}
Influence Maximization; Viral Marketing; Influence Barricade Model; Theory; Algorithm; Non-submodular
\end{IEEEkeywords}

\section{Introduction}
\label{intro}

The development of high-end portable devices and utmost prevalence of social networks around the globe have drastically increased the speed and frequency of interactions among individuals. Unlike communication networks, social networks provide a unique substrate through which social behaviors, e.g., the adoption of innovations, trends, social awareness, etc., can also propagate. In particular, by leveraging the social impact, individuals with certain social behaviors can affect their friends, and the influenced ones can further affect their friends, etc. The spread of such social influence is referred to as \emph{influence diffusion} in social networks. As one canonical application of influence diffusion, viral marketing takes advantage of the ``word-of-mouth'' effect \cite{Domingos01KDD} to promote the vast spread of product awareness and adoption in a viral replication manner. Since the influence among friends is generally more reliable than third-party media (e.g., TV, radio, and billboard), viral marketing demonstrates substantial advantages than traditional commercial campaigns. Motivated by such influence-diffusion applications, \emph{influence maximization} emerges as a fundamental research issue. Specifically, influence maximization explores a computationally efficient way to select a small subset of the most influential nodes from a social network so that the selected nodes can trigger the maximal influence diffusion. These selected nodes are called \emph{seed nodes}, which act as initial influential sources; moreover, seed  selection is subject to practical constraints, e.g., budget for free samples, and thus only a small subset is allowed to serve as seed nodes.

For influence maximization, it is first investigated as an algorithmic problem by \cite{Domingos01KDD,Richardson02KDD} using statistical models and parameter estimation. Following their work, influence maximization is formulated as a discrete optimization problem in \cite{Kempe03KDD}, where two diffusion models, Independent Cascade model (IC) and Linear Threshold model (LT), are explored. Under these two models, influence maximization is reduced to an easier combinatorial problem where the objective function is submodular.\footnote{Function $f(\cdot)$ mapping a subset of $U$ to a real number is \emph{submodular} if for every $X,Y\subseteq U$ with $X\subseteq Y$ and every $u\in U\setminus Y$, $f(X\cup\{u\})-f(X)\geq f(Y\cup\{u\})-f(Y)$ always holds.}\label{footnote1} For such submodular influence maximization problems, simple hill-climbing greedy algorithm performs well with guaranteed $(1-1/e)$ approximation ratio \cite{Nemhauser78MathProg,Kempe03KDD}. As such, a large amount of research works based on IC/LT models are stimulated thereafter, e.g., algorithms with improved complexity in \cite{Leskovec07KDD,Chen09KDD,Chen10KDD,Jiang11AAAI,Borgs14SODA} and optimization under more constraints in\cite{Ma08CIKM,Zhang13ICDCS,Zhu16INFOCOM}.
Due to the limitation of IC/LT in capturing various social effects in influence propagation, other complicated diffusion models are proposed, e.g., \cite{Lu15VLDB,Lin17ICDE,Anshelevich15AAMAS,Zhang14KDD,Karampourniotis2017,Yang12CIKM,Chen09SDM}, which are, however, no longer submodular. For influence maximization under these non-submodular models, one common approach is to find other submodular models that can approximate and/or bound the objective function in the original problem; the derived results associated with these relaxed models are then used as the estimated solution to the original non-submodular problem. In particular, \cite{Lu15VLDB} establishes a  Sandwich Approximation (SA) strategy that bounds the original problem from both sides (lower/upper bounds) by two submodular functions, and proves the approxiamation ratio of the solution derived from these submodular problems. However, depending on how the bounding submodular functions are chosen, the corresponding result can be arbitrarily worse than the optimal solution. Therefore, there is still a lack of efficient seed selection strategies that can directly tackle non-submodular influence maximization problems. In this regard, we target to investigate efficient algorithms specifically designed to solve non-submodular problems by leveraging network intrinsic attributes, i.e., graphical properties.

To this end, the seed selection approaches need to be developed under diffusion models that are non-submodular and sufficiently generalized to cover more social effects. In the classical IC model, individuals independently influence their neighbors with a predetermined probability, and the correlated influence from neighbors is ignored. Though the LT model considers certain aspects of the correlated neighboring influence, LT cannot capture general hard (non-statistical) influence requirements that must be met to enable the influence diffusion in some real networks. Therefore,\cite{Karampourniotis2017,Yang12CIKM,Chen09SDM} extend LT to consider these hard constraints, which results in the problem to be non-submodular. Broadly speaking, LT and its extensions all fall into the realm of threshold-based models. Specifically, under the threshold-based model, each node associates with a requirement of being influenced; only when the accumulated impact from a node's neighbors exceeds its influence requirement can this node be influenced. However, the parameter settings in these models are highly constrained (i.e., the range of social influence and thresholds). As such, in this paper, we consider one generic threshold-based non-submodular model, called \emph{influence barricade model}, where the influence requirement is a hard constraint (not a random variable) that must be satisfied on all cases for the spread of influence. More importantly, no parameter constraints (e.g., social impacts and influence requirements) in \cite{Kempe03KDD,Karampourniotis2017,Yang12CIKM,Chen09SDM} are imposed to this model.
Such generic model can capture any barricade-like social effect that is abundant \cite{Granovetter78,Aral11Mgmt} in real life. Moreover, this barricade model incorporates the case where there may exist people who are too conservative to be influenced by their friends, which is generally ignored in other threshold-based models.\looseness=-1

To develop efficient solutions under such generic model, we first establish theoretical conditions in terms of network graphical properties to reveal the relationships between the optimal seed sets in two networks that differ by only a few nodes/edges; the established theories indicate under what conditions these two graphs maintain the same (or similar) optimal seed set. Based on such fundamental understanding, we then develop efficient seed selection algorithms for non-submodular influence maximization problems. The basic idea of our proposed algorithms is to iteratively remove nodes, one at a time, with each step trying to guarantee that the remaining network and the original network have (almost) the same optimal seed set using our theoretical conditions. This removal process proceeds until the selection budget is reached. Evaluations in both synthetic and real Facebook/Twitter networks confirm the orders of magnitude improvement of our graph-based algorithms, especially in large networks.

It is worth noting that the main goal of this paper is to develop an efficient solution to non-submodular influence maximization problems. For this goal, we investigate it under a generic diffusion model in the family of threshold-based models. Investigation under other classes of diffusion models is out of the scope of this paper. Nevertheless, our graph-based approaches can shed light on how to use network graphical properties to solve other non-submodular problems.
\subsection{Further Discussions on Related Work}
\label{subsect:relatedWork}

For the basic hill-climbing greedy solution \cite{Kempe03KDD}, the lack of scalability makes it inapplicable to large networks. Therefore, computationally efficient algorithms are developed in \cite{Leskovec07KDD,Chen09KDD,Chen10KDD,Jiang11AAAI,Borgs14SODA}. These algorithms leverage the fact that the influential property associated with most nodes remain unchanged even new nodes are selected as seeds, and thus there is no need to compute these properties again in the following steps.
These efficient algorithms, however, are developed under the IC model, thus not applicable to non-submodular influence maximization problems. Imposing other constraints, \cite{Liu12ICDM,Chen12AAAI,Lu16TMC,Zhu16INFOCOM} investigate how to select the seed set while satisfying the time or cost constraint under the IC model.
Moreover, individuals in social networks may also experience negative effect from neighbors, which are examined in \cite{Ma08CIKM,Zhang13ICDCS}. However, works in \cite{Ma08CIKM,Zhang13ICDCS} are restricted to independent influence cascade or strict parameter values. Another line of research related to our problem is influence maximization in the presence of other competitors \cite{Carnes07ICEC,Lin15Performance} under the assumption of knowing the exact locations or location distributions of competitors' seed nodes. Regarding such competitive influence maximization problems, again they focus on the case where the problem is submodular, thus not widely applicable.

When more social effects are considered, the problems are converted to non-submodular. Under non-submodular models, \cite{Zhang13ICDCS,Zhang14KDD} reuse a pure greedy-search-based algorithm, which is not specific to submodularity. To improve/quantify the algorithm performance, \cite{Lu15VLDB,Anshelevich15AAMAS} approximate the original problem by simpler submodular functions. In addition, Sandwich Approximation is employed in \cite{Lu15VLDB,Lin17ICDE} to bound the original problems by tractable submodular functions. However, such submodularity-approximation-based approaches are sensitive to the submodular functions selected for estimation, and thus there is no concrete performance guarantees. By contrast, we propose state-of-the-art approaches from both theoretical and algorithmic perspectives to directly and specifically address the original non-submodular influence maximization problem.

\subsection{Summary of Contributions}

We study non-submodular influence maximization from the perspective of both theories and algorithm designs. Our contributions are four-fold:

\emph{1)} We prove the influence maximization problem under the influence barricade model is NP-hard and non-submodular.

\emph{2)} We establish fundamental theories that reveal to what extent network graphical changes can affect the optimal seed set, and more importantly, under what conditions minor graphical changes can maintain the optimal seed set.

\emph{3)} We develop superior algorithms based on the theoretical results in \emph{2)} to iteratively reduce the seed selection scope until the optimization goal or the selection budget is reached.

\emph{4)} We evaluate the proposed algorithm in both synthetic networks and real Facebook/Twitter datasets, which confirms the exceptional efficiency of our solutions.

Note that our goal is fundamental understanding and efficient algorithm development that are specifically designed for non-submodular influence maximization problems. To this end, we 
assume that all parameters in the proposed model can be extracted from a network targeted for influence maximization (e.g., using the approaches in \cite{Domingos01KDD,Richardson02KDD}), and are available as input to our algorithms. Refined methods for inferring these parameters are beyond the scope of this paper.\looseness=-1

\vspace{.5em}
The remainder of the paper is organized as follows. Section~\ref{Sect:ProblemFormulation} formulates the problem; corresponding research challenges are discussed in Section~\ref{sect:ResearchChallenges}. Main theoretical results for non-submodular influence maximization are presented in Section~\ref{sect:TheoFoundations}, based on which seed selection algorithms are developed in Section~\ref{sect:AlgSeedSetSelection}. Experiments under both synthetic and real networks are conducted in Section~\ref{sect:evaluations}. Section~\ref{Sect:Conclusion} concludes the paper.

\section{Problem Formulation}
\label{Sect:ProblemFormulation}

In this section, we first introduce the generic threshold-based diffusion model. We then formally formulate our influence maximization problem and state our research objective.

\vspace{-.3em}
\subsection{Influence Diffusion in Social Networks}
\label{subsec:networkModel}
Let directed graph $\mathcal{G}=(V,L)$ represent a social network, 
where $V$ ($L$) denotes the set of nodes (directed edges). $\mathcal{G}$ can be either a connected or disconnected graph; our proposed algorithms are independent of the network connectivity. In $\mathcal{G}$, $\small\overrightarrow{uv}\normalsize\in L$ is a directed edge starting from node $u$ and pointing to node $v$; in addition, edge $\small\overrightarrow{uv}\normalsize$ associates with a positive weight\footnote{For any edge \emph{not} in $L$, the associated weight is defined as $0$.\looseness=-1} $W_{\overrightarrow{uv}}$ (i.e., $W_{\overrightarrow{uv}}>0$, which can be \emph{different} for different edges), which represents the quantified social influence of node $u$ on node $v$. In $\mathcal{G}=(V,L)$, regarding $u\in V$, $v$ is called an \emph{out-neighbor} (or \emph{in-neighbor}) of $u$ if there exists $\small\overrightarrow{uv}\normalsize\in L$ (or there exists $\small\overrightarrow{vu}\normalsize\in L$). For each node in $\mathcal{G}$, its status is either \emph{active} or \emph{inactive}. A node $u\in V$ is active if $u$ adopts the innovations, trends, awareness, etc., that are being propagated in the network; otherwise, $u$ is inactive. In this paper, we consider the progressive diffusion process, where node status can switch from inactive to active, but status cannot change in the opposite direction. In social network $\mathcal{G}=(V,L)$, let $S$ ($S\subseteq V$) be the set of initial active nodes, referred to as \emph{seed nodes}, which act as the original influence sources (e.g., nodes selected by a company for a product campaign). Based on such seed set, inactive nodes may become active in subsequent time steps via influence diffusion.


\subsection{Influence Barricade Model}
\label{subsec:InfluenceBarricadeModel}

In this paper, we consider the following generic threshold-based diffusion model, called \emph{influence barricade model}. Suppose $\forall u\in V$, $u$ associates with a fixed non-negative influence threshold $b_u$ (i.e., $b_u\geq 0$, which can be \emph{different} for different nodes), called the \emph{barricade factor}, which determines the difficulty level of $u$ to be influenced by its neighbors. In particular, given the network state at time $t$, for any inactive node $u$, $u$ becomes active at time $t+1$ if and only if the sum weight from all its active in-neighbors is at least $b_u$ (i.e., $\sum_{v\text{ is an active in-neighbor of }u}W_{\overrightarrow{vu}}$ $\geq b_u$); such diffusion cascade proceeds as a \emph{discrete}-time process. In this model, no constraints are imposed on the values of barricade factors or the edge weights in the social network. Here, the barricade factor with respect to a node represents the hard and uncompromisable requirement that must be satisfied by its active neighbors so as to influence this node. Such barricade-like diffusion effect has been shown in \cite{Granovetter78,Aral11Mgmt,Bakshy12WWW} to be abundant in real life (e.g., diffusion in knowledge networks). Specifically, a barricade factor quantifies a node's unwillingness in adopting innovations or other social phenomena; moreover, such social unwillingness is unlikely to change within a short period of time, and thus assumed to be a fixed parameter.\looseness=-1

\vspace{.5em}
\emph{Discussions:} At first sight, the influence barricade model is similar to the models\footnote{In \cite{Yang12CIKM}, when the per-node preference score is set to $0$, the corresponding diffusion model is reduced to a simple threshold-based model.} in \cite{Karampourniotis2017,Yang12CIKM,Chen09SDM} in the way of how influence cascades in weighted directed graphs with nodes exhibiting non-negative influence thresholds. Nevertheless, in the influence barricade model, we do not impose any constraints on the maximum value of edge weights or barricade factors. By contrast, the threshold parameter is limited to $[0,1]$ (i.e., the same as that in the LT model) in \cite{Karampourniotis2017,Yang12CIKM}, and $[0,d(v)]$ for node $v$ ($d(v)$: the number of neighbors of node $v$) in \cite{Chen09SDM}. Therefore, each node under these parameter constraints can be influenced by their neighbors with a positive probability.
By contrast, under our barricade model, there may exist a non-negligible number of nodes who cannot be influenced by their neighbors (i.e., corresponding barricade factors are large), and thus these nodes have to be selected as seed nodes if they are targeted to be active. Therefore, the influence barricade model is more general in the realm of threshold-based diffusion models.\looseness=-1

\subsection{Influence Maximization Problem}
With the influence barricade model, we now formally define the influence maximization problem. Let $A(S)$ denote the set of all active nodes ($S\subseteq A(S)$) that are influenced by seed set $S$ at the moment when no other nodes can become active even if more time is given; this particular moment is called \emph{the end of the influence process} in the sequel. Let $k$ be the maximum number of seed nodes that are allowed to be selected from a network, i.e.,
$|S|\leq k$, and $\sigma(S):=|A(S)|$ the total number of active nodes at the end of the influence process. Our objective is to maximize the influence by selecting a constrained number of seed nodes, as stated below.

\vspace{.5em}
\textbf{Objective:} Given graph $\mathcal{G}=(V,L)$, select a set of seed nodes $S$ from $V$ with $|S|\leq k$ such that $\sigma(S)$ is maximized under the influence barricade model.
\vspace{.5em}

Let seed set $S^*_k$ denote the optimal solution to the above problem when up to $k$ seed nodes are allowed to be selected from $V$. We aim to determine $S^*_k$ under any $k$. Regarding the global network status, we call the network under one seed set achieves \emph{full influenceability} if every node is active at the end of the influence process; otherwise, we call it \emph{partial influenceability}. With this concept, let $S^*_{\mathcal{G}}$ be a minimum seed set required for full influenceability in $\mathcal{G}=(V,L)$, i.e., $\sigma(S^*_{\mathcal{G}})=|V|$ and $\forall S\in\{Q\subseteq V: |Q|<|S^*_{\mathcal{G}}|\}$ $\sigma(S)<|V|$. Then we can further define $S^*_k$ as $S^*_k=S^*_{\mathcal{G}}$ if $k\geq |S^*_{\mathcal{G}}|$ since $\sigma(\cdot)$ already achieves the maximum when $k= |S^*_{\mathcal{G}}|$. Moreover, $S^*_k$ generally is not unique; see Section~\ref{subsec:Example} for examples.
%

\subsection{Illustrative Example}
\label{subsec:Example}

\begin{figure}[tb]
\vspace{-.3em}
\centering
\includegraphics[width=3.2in]{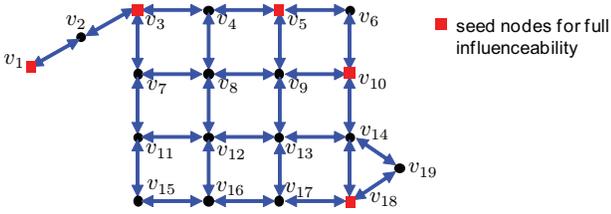}
\vspace{-1em}
\caption{Seed selection in a sample network (weight is $1$ for all edges; barricade factor is $2$ for all nodes; $v_1$ must be selected as a seed node for full influenceability).} \label{fig:Illustrative_Example}
\vspace{-.5em}
\end{figure}

Fig.~\ref{fig:Illustrative_Example} shows a sample bidirectional network with $19$ nodes. Suppose the weight is $1$ for each edge and the barricade factor is $2$ for each node. Then the optimal seed set for full influenceability is $S^*_{\mathcal{G}}=\{v_1,v_3,v_5,v_{10},v_{18}\}$. Let $A_{t}(S)$ be the set of active nodes at time step $t$ under seed set $S$, where $A_0(S)=S$. Then $A_{0}(S^*_{\mathcal{G}})=S^*_{\mathcal{G}}$,
$A_{1}(S^*_{\mathcal{G}})=A_{0}(S^*_{\mathcal{G}})\cup\{v_2,v_4,v_6,v_{14}\}$,
$A_{2}(S^*_{\mathcal{G}})=A_{1}(S^*_{\mathcal{G}})\cup\{v_9,v_{19}\}$,
$A_{3}(S^*_{\mathcal{G}})=A_{2}(S^*_{\mathcal{G}})\cup\{v_8,v_{13}\}$,
$A_{4}(S^*_{\mathcal{G}})=A_{3}(S^*_{\mathcal{G}})\cup\{v_7,v_{12},v_{17}\}$,
$A_{5}(S^*_{\mathcal{G}})=A_{4}(S^*_{\mathcal{G}})\cup\{v_{11},v_{16}\}$,
$A_{6}(S^*_{\mathcal{G}})=A_{5}(S^*_{\mathcal{G}})\cup\{v_{15}\}$; therefore, $\sigma(S^*_{\mathcal{G}})=19$. In $S^*_{\mathcal{G}}$, node $v_1$ must be selected as a seed node for full influenceability, as $v_1$ cannot be influenced even if all other nodes in the network are active. Given $S^*_{\mathcal{G}}$,
we know that $S^*_k=\{v_1,v_3,v_5,v_{10},v_{18}\}$ for $k\geq 5$. For other values of $k$, we have $S^*_1=\{v_{18}\}$ with $A(S^*_1)=\{v_{18}\}$, $S^*_2=\{v_{13},v_{18}\}$ with $A(S^*_2)=\{v_{13},v_{14},v_{17},v_{18},v_{19}\}$, $S^*_3=\{v_8,v_{13},v_{18}\}$ with $A(S^*_3)=\{v_8,v_9,v_{10},v_{12},v_{13},v_{14},v_{16},v_{17},v_{18},v_{19}\}$, and $S^*_4=\{v_3,v_8,v_{13},v_{18}\}$ with $A(S^*_4)=\{v_i:i=3,4,\ldots,19\}$. Note that the optimal solution is not unique, e.g., seed set $\{v_4,v_9,v_{14},v_{15}\}$ is also optimal when $k=4$, and seed set $\{v_1,v_3,v_8,v_{13},v_{18}\}$ also achieves full influenceability using $5$ seed nodes.\looseness=-1

\section{Research Hardness}
\label{sect:ResearchChallenges}

Under the influence barricade model, we prove that the influence maximization problem is NP-hard, detailed below.

\begin{theorem}\label{thm:NPhard}
The influence maximization problem, i.e., maximizing $\sigma(S)$ subject to $|S|\leq k$, is NP-hard under the influence barricade model.
\end{theorem}
\begin{proof}
This can be proved by the reduction from the \emph{Vertex Cover} problem, which is known to be NP-hard. Consider one instance of the Vertex Cover problem defined on undirected graph $\mathcal{G}=(V,L)$ and $k$, where the goal is to know whether there exists a size-$k$ set $S\in V$ such that each edge in $L$ has at least one endpoint in $S$. Given this Vertex Cover problem, an influence maximization problem is defined on a weighted directed graph $\mathcal{G}'=(V,L')$, where (i) $\mathcal{G}'$ and $\mathcal{G}$ have the same set of nodes, (ii) each edge in $L$ is changed to bidirectional edges with weights being $1$ in both directions and included in $L'$, and (iii) the barricade factor $b_v$ is the number of in-neighbors of $v$ in $\mathcal{G}'$. We prove that $\exists$size-$k$ set $S$ in $\mathcal{G}'$ with $\sigma(S)=|V|$ if and only if $S$ is a vertex cover in $\mathcal{G}$.

Suppose $S$ is a vertex cover in $\mathcal{G}$. Then for each node $v\in V\setminus S$, the sum weights of $v$'s in-neighbors is the same as the barricade factor $b_v$ due to the way that $\mathcal{G}'$ is constructed; therefore, $\sigma(S)=|V|$. Next, suppose $\exists$ a seed set $S$ in $\mathcal{G}'$ with $\sigma(S)=|V|$ but $S$ is not a set cover in $\mathcal{G}$. Then $\exists$ an edge $uv\in L$ that is not incident to any nodes in $S$. Since the barricade factor for both $u$ and $v$ is the corresponding number of in-neighbors in $\mathcal{G}'$, we have that $u$ becoming active only if $v$ is already active, and vice versa. Therefore, $u$ and $v$ cannot be both active even if all nodes in $S$ are active. This contradicts the assumption that $\sigma(S)=|V|$.

Hence, the above two problems are equivalent, and thus our influence maximization problem is NP-hard.
\end{proof}

Therefore, we need to develop heuristic solutions. In addition, the influence barricade model causes the influence maximization problem to be non-submodular, as stated below.

\begin{claim}
\label{clm:non-submodular}
The influence maximization problem under the influence barricade model is \emph{not} submodular.
\end{claim}

\begin{proof}
We prove by a counter example. Under the network settings in Fig.~\ref{fig:Illustrative_Example}, let $S_1=\{v_8,v_{13}\}$ and $S_2=\{v_3,v_8,v_{13}\}$. Then $\sigma(S_1)=4$ and $\sigma(S_2)=9$; moreover, $\sigma(S_1\cup\{v_{18}\})=10$ and $\sigma(S_2\cup\{v_{18}\})=17$. Obviously, $S_1\subset S_2$, $v_{18}\notin S_2$, and $\sigma(S_1\cup\{v_{18}\})-\sigma(S_1)<\sigma(S_2\cup\{v_{18}\})-\sigma(S_2)$, which violates the definition of submodularity in footnote~1.\looseness=-1
\end{proof}

Regarding such non-submodular problem, we propose state-of-the-art algorithms by leveraging network graphic properties to identify structurally critical nodes.

\section{Theoretical Foundations for \\Seed Set Selection}
\label{sect:TheoFoundations}

The algorithm development for efficient seed selection requires fundamental understanding of the linkage between network properties and node importance levels of being seeds.
In this regard, we establish design principles from the theoretical perspective, focusing on identifying crucial network properties that may affect the seed selection for full influenceability. The significance of these properties is that they provide a concrete support for the design goal of both full and partial influenceability (see Section~\ref{sect:AlgSeedSetSelection} for details). In particular, we investigate how minor changes to the network (e.g., add/remove a(n) node/edge) may affect set $S^*_{\mathcal{G}}$ (recall that $S^*_{\mathcal{G}}$ denotes the minimum seed set for full influenceability in $\mathcal{G}$). The idea for our following-up algorithm design is that if a certain sequence of minor changes is proved to have little or no impact on $S^*_{\mathcal{G}}$, then we are able to narrow $S^*_{\mathcal{G}}$ down to the nodes in the remaining network. Since node changes generally also incur edge changes, we therefore first study the impact of minor edge changes on the optimal seed set selection.

\begin{theorem}
\label{thm:addOneEdge}
Let $\mathcal{G}=(V,L)$ and $\mathcal{G}'=(V,L\cup\{\overrightarrow{v_1v_2},\overrightarrow{v_2v_1}\})$ with the corresponding minimum seed set for full influenceability being $S^*_{\mathcal{G}}$ and $S^*_{\mathcal{G}'}$ respectively, where $v_1,v_2\in V$ and $\overrightarrow{v_1v_2},\overrightarrow{v_2v_1}\notin L$. Then $0\leq |S^*_{\mathcal{G}}|-|S^*_{\mathcal{G}'}|\leq 1$.
\end{theorem}
\begin{proof}
It is obvious that $|S^*_{\mathcal{G}}|\geq |S^*_{\mathcal{G}'}|$ as $\mathcal{G}$ is a subgraph of $\mathcal{G}'$.
Moreover, at most one of $\small\overrightarrow{v_1v_2}\normalsize$ and $\small\overrightarrow{v_2v_1}\normalsize$ is useful in achieving full influenceability in $\mathcal{G}'$ as $v_1$ and $v_2$ cannot mutually influence each other. Without loss of generality, suppose $\small\overrightarrow{v_1v_2}\normalsize$ is used to influence $v_2$ for achieving full influenceability. Then when $\small\overrightarrow{v_1v_2}\normalsize$ and $\small\overrightarrow{v_2v_1}\normalsize$ are removed from $\mathcal{G}'$ (i.e., we get $\mathcal{G}$), to maintain the full influenceability in $\mathcal{G}$, the worst case is that $v_2$ is selected as a seed node while all other seed nodes remain unchanged. Thus, $|S^*_{\mathcal{G}}|-|S^*_{\mathcal{G}'}|\leq1$.
\end{proof}

Following similar arguments in the proof of Theorem~\ref{thm:addOneEdge}, the following corollary also holds.

\begin{corollary}
\label{cor:addOneEdge}
Let $\mathcal{G}=(V,L)$ and $\mathcal{G}'=(V,L\cup\{\overrightarrow{v_1v_2}\})$ with the corresponding minimum seed set for full influenceability being $S^*_{\mathcal{G}}$ and $S^*_{\mathcal{G}'}$ respectively, where $v_1,v_2\in V$ and $\overrightarrow{v_1v_2}\notin L$. Then $0\leq |S^*_{\mathcal{G}}|-|S^*_{\mathcal{G}'}|\leq 1$.
\end{corollary}

Theorem~\ref{thm:addOneEdge} and Corollary~\ref{cor:addOneEdge} consider a simple case where $\mathcal{G}$ and $\mathcal{G}'$ differ by only one or two edges between a pair of nodes.
Since $\mathcal{G}$ is a subgraph of $\mathcal{G}'$, the additional edges in $\mathcal{G}'$ do not have negative impact on the spread of influence, and thus it is obvious that $S^*_{\mathcal{G}}$ also achieves full influenceability in $\mathcal{G}'$. Nevertheless, with additional edges in $\mathcal{G}'$, it is possible that the size of the minimum seed set for full influenceability in $\mathcal{G}'$ can be reduced, i.e., $|S^*_{\mathcal{G}'}|< |S^*_{\mathcal{G}}|$. However, Theorem~\ref{thm:addOneEdge} and Corollary~\ref{cor:addOneEdge} show that such reduction is upper bounded by $1$, i.e., additional edges between a pair of nodes can reduce the size of the minimum seed set for full influenceability by at most $1$. Therefore, the significance of Theorem~\ref{thm:addOneEdge} and Corollary~\ref{cor:addOneEdge} is that they quantify the impact of edge changes by providing lower and upper bound with the gap being only $1$. More importantly, these results are independent of the values of $\{b_u\}_{u\in V}$ and $\{W_{\overrightarrow{uv}}\}_{\overrightarrow{uv}\in L}$, thus applicable to networks under any parameter settings. Next, based on these results, we consider the case of node changes, which is a more complicated case as it is also accompanied by the changes of neighboring edges. For ease of presentation, we first give two definitions.\looseness=-1

\begin{definition}
\label{def:deficient}
\emph{\\ \indent 1)}
Node $u$ with $b_u>\sum_{w\in N_{\text{in}}}W_{\overrightarrow{wu}}$, where $N_\text{in}$ is the set of all in-neighbors of $u$ in network $\mathcal{G}$, is an \emph{influence deficient node}, for which $u\in S^*_{\mathcal{G}}$ always holds.

2) An edge is \emph{redundant} if the removal of which does not affect the size of the minimum seed set for full influenceability, i.e., edges $L'$ in $\mathcal{G}=(V,L)$ are redundant if and only if $|S^*_{\mathcal{G}}|=|S^*_{\mathcal{G}'}|$, where $\mathcal{G}'=\{V,L\setminus L'\}$.
\end{definition}

With these definitions, we derive lower/upper bounds to quantify the maximum impact of single node changes on the optimal seed set for full influenceability as follows, where $\mathcal{V}(\mathcal{G}')$ denotes the set of nodes in graph $\mathcal{G}'$.

\begin{corollary}
\label{cor:addOneNode}
Given graph $\mathcal{G}=(V,L)$, suppose an extra node $v$ ($v\notin V$) is added and connected to a set of nodes $N$ in $\mathcal{G}$, thus forming a new graph $\mathcal{G}'$. Let $S^*_{\mathcal{G}}$ and $S^*_{\mathcal{G}'}$ denote the minimum seed set for full influenceability in $\mathcal{G}$ and $\mathcal{G}'$ respectively. Then $\max(|M|,|S^*_{\mathcal{G}}|+1-|N|)\leq |S^*_{\mathcal{G}'}|\leq |S^*_{\mathcal{G}}|+1$, where $M=\{w\in \mathcal{V}(\mathcal{G}'): b_w>\sum_{z\in \mathcal{V}(\mathcal{G}')} W_{\overrightarrow{zw}}\}$ (i.e., influence deficient nodes in $\mathcal{G}'$).
\end{corollary}
\begin{proof}
Let $\mathcal{G}''=(V\cup\{v\},L)$, which requires minimum seed set $S^*_{\mathcal{G}''}$ for full influenceability. Then $|S^*_{\mathcal{G}''}|-|S^*_{\mathcal{G}'}|\leq |N|$ by Theorem~\ref{thm:addOneEdge}. Since $v$ does not have any neighbors in $\mathcal{G}''$, $S^*_{\mathcal{G}''}=S^*_{\mathcal{G}}\cup \{v\}$; therefore, $|S^*_{\mathcal{G}}|+1-|S^*_{\mathcal{G}'}|\leq |N|$. In addition, $S^*_{\mathcal{G}}\cup \{v\}$ also achieves full influenceability in $\mathcal{G}'$, and influence deficient nodes must be seeds for full influenceability. Thus, $\max(|M|,|S^*_{\mathcal{G}}|+1-|N|)\leq |S^*_{\mathcal{G}'}|\leq |S^*_{\mathcal{G}}|+1$.
\end{proof}

Adding a node to $\mathcal{G}$, it is trivial that $S^*_{\mathcal{G}'}$ in the resulting graph $\mathcal{G}'$ must include all influence deficient nodes for full influenceability, i.e., $|S^*_{\mathcal{G}'}|\geq |M|$ in Corollary~\ref{cor:addOneNode}. On the other hand, Corollary~\ref{cor:addOneNode} also shows that $|S^*_{\mathcal{G}'}|\geq |S^*_{\mathcal{G}}|+1-|N|$, which is a lower bound depending on the degree of the added node (but independent of $\{b_u\}_{u\in V}$, $\{W_{\overrightarrow{uv}}\}_{\overrightarrow{uv}\in L}$, and the redundancy of edges). We next develop more concrete conditions to characterize the relationships between $S^*_{\mathcal{G}}$ and $S^*_{\mathcal{G}'}$, with the detailed proof in the Appendix.

\begin{theorem}
\label{THM:ADDV}
Given graph $\mathcal{G}=(V,L)$, suppose node $v$ ($v\notin V$) is added and connected to a set of nodes $N$ in $\mathcal{G}$ ($N\subseteq V$, $N=N_\text{in}\cup N_\text{out}$, $N_\text{in}$: all in-neighbors of $v$, $N_\text{out}$: all out-neighbors of $v$), thus forming a new graph $\mathcal{G}'=(V\cup\{v\}, L\cup L' )$, where $L'=\{\overrightarrow{wv}:w\in N_\text{in}\}\cup \{\overrightarrow{vw}:w\in N_\text{out}\}$ and each edge in $L'$ associates with a positive weight. Let $Q$ be the largest subset of $N$ ($Q\subseteq N$) that satisfies the following three conditions.
\begin{enumerate}[(a)]
  \item $\exists$ a minimum seed set $S^*_{\mathcal{G}}$ for full influenceability in $\mathcal{G}$ with $Q\subset S^*_{\mathcal{G}}\subseteq V$;
  \item let $Z:=\{z:z\notin Q,z$ is a neighbor of at least one node in $Q\}$. All edges between $Q$ and $Z$ (in either direction) are redundant in $\mathcal{G}$;
  \item $\forall q\in Q$, $\sum_{z\in Z}W_{\overrightarrow{zq}}<b_q\leq W_{\overrightarrow{vq}}+\sum_{z\in Z}W_{\overrightarrow{zq}}$.
\end{enumerate}
Then regarding the minimum seed set $S^*_{\mathcal{G}'}$ for full influenceability in $\mathcal{G}'$, the following results hold.
\begin{enumerate}[(1)]
  \item If $\sum_{w\in N'_{\text{in}}} W_{\overrightarrow{wv}}<b_v\leq \sum_{w\in N_{\text{in}}} W_{\overrightarrow{wv}}$ (for any $N'_{\text{in}}\subset N_{\text{in}}$), $N_\text{out}\subseteq N_\text{in}$, and $W_{\overrightarrow{vw}}<b_w$ (for any $w\in N_\text{out}$), then the necessary and sufficient condition for $|S^*_{\mathcal{G}'}|<|S^*_{\mathcal{G}}|$ is that $|Q|>1$. Specifically,

$\left\{ \begin{aligned}
         &S^*_{\mathcal{G}'}=(S^*_{\mathcal{G}}\setminus Q)\cup \{v\}, &|Q|>1\\
         &S^*_{\mathcal{G}'}=S^*_{\mathcal{G}}, &|Q|\leq 1
         \end{aligned} \right.$.
  \item If $\sum_{w\in N_{\text{in}}} W_{\overrightarrow{wv}}<b_v$ and $W_{\overrightarrow{vw}}<b_w$ (for any $w\in N_\text{out}$), then the necessary and sufficient condition for $|S^*_{\mathcal{G}'}|<|S^*_{\mathcal{G}}|$ is that $|Q|>1$. Specifically, $S^*_{\mathcal{G}'}=(S^*_{\mathcal{G}}\setminus Q)\cup \{v\}$ for any $|Q|\geq 0$.
\end{enumerate}
\end{theorem}

For graphs $\mathcal{G}=(V,L)$ and $\mathcal{G}'=(V\cup\{v\}, L\cup L' )$, Theorem~\ref{THM:ADDV} states under what conditions and to what extent the size of $S^*_{\mathcal{G}'}$ varies. All these conditions are expressed in terms of a newly defined set $Q$. Intuitively, $Q\subset S^*_{\mathcal{G}}$ contains all seed nodes whose role in influence diffusion can be replaced by a single seed node $v$. Specifically, in condition (a), nodes in $Q$ along with other seed nodes constitute one optimal solution for full influenceability (note that the optimal solution $S^*_{\mathcal{G}}$ is not unique). In condition (b), when $Q\subset S^*_{\mathcal{G}}$, the influence from/to $Q$ on its neighbors (i.e., set $Z$ in Theorem~\ref{THM:ADDV}) can be ignored as all edges between $Q$ and $Z$ are redundant. This suggests that if selecting $v$ as a seed can activate all nodes in $Q$, then nodes in $Q$ do not need to be seeds as seed $v$ and other seed nodes suffice to maintain full influenceability. Finally, to ensure that the impact from $Q$ can be replaced by a single seed node $v$, the inequality in condition (c) should be satisfied. Then based on set $Q$, two cases are considered. In Case (1), the neighbors of $v$ are just sufficient to activate $v$, i.e., no subset of these neighbors is sufficient. In this case, Theorem~\ref{THM:ADDV} states that the size of $S^*_{\mathcal{G}'}$ will not increase as $S^*_{\mathcal{G}}$ also achieves full influenceability in $\mathcal{G}'$, and thus $S^*_{\mathcal{G}'}$ is at most as large as $S^*_{\mathcal{G}}$. Moreover, when $|Q|>1$, $|S^*_{\mathcal{G}'}|$ is reduced, i.e., $S^*_{\mathcal{G}'}=(S^*_{\mathcal{G}}\setminus Q)\cup \{v\}$. While for Case (2), we consider the case where the neighbors of the added node $v$ are no longer sufficient to activate $v$, i.e., $v$ is an influence deficient node, for which we must have $v\in S^*_{\mathcal{G}'}$. Then Theorem~\ref{THM:ADDV} presents a uniform expression of the optimal set $S^*_{\mathcal{G}'}$ under any $Q$. Moreover, due to different attributes of the neighbors of $v$, Theorem~\ref{THM:ADDV} shows that $|S^*_{\mathcal{G}'}|\leq |S^*_{\mathcal{G}}|$ in Case (1), whereas $|S^*_{\mathcal{G}'}|\leq |S^*_{\mathcal{G}}|+1$ in Case (2).

\emph{Discussions:} One uncovered case by Theorem~\ref{THM:ADDV} is that there may exist a subset $N'\subset N$ whose joint impact is sufficient to activate $v$. In this case, when $v$ is active, it can further influence nodes in $N\setminus N'$; then the corresponding $|S^*_{\mathcal{G}'}|$ may be significantly reduced. This is because the connectivity of graph $\mathcal{G}'$ is considerably enhanced when $|N|$ is large,
and thus the influence requirement (barricade factors) can be satisfied more easily. Furthermore, in Theorem~\ref{THM:ADDV}, we have conditions $N_\text{out}\subseteq N_\text{in}$ and $W_{\overrightarrow{vw}}<b_w$ (for any $w\in N_\text{out}$), which essentially are the constraints on the influence of $v$ to other nodes in the network. Without such constraints, $v$ may influence its neighbors directly without the assistance of any other seeds in the network (e.g., when the weight from $v$ to its neighbors is larger than the corresponding barricade factors). Recall that our goal in this section is to identify what minor network changes can incur little/no impact on the optimal seed set for full influenceability. The cases not considered by Theorem~\ref{THM:ADDV} generally cause undesirable changes to the optimal seed set, i.e.,  $S^*_{\mathcal{G}'}$ and $S^*_{\mathcal{G}}$ may not have any relationships. Thus, only cases that can express $S^*_{\mathcal{G}'}$ in terms of $S^*_{\mathcal{G}}$ are considered in Theorem~\ref{THM:ADDV}.\looseness=-1

Besides being of theoretical value, these theorems and corollaries can also guide our algorithm design. Specifically, using these results, we are able to identify which nodes are critical in affecting the optimal seed set for full influenceability. Based on such understanding, we develop efficient seed selection algorithms in the next section.

\section{Seed Set Selection Algorithms}
\label{sect:AlgSeedSetSelection}

In this section, we build efficient algorithms for seed set selection. According to Theorem~\ref{thm:addOneEdge}, Corollaries~\ref{cor:addOneEdge} and \ref{cor:addOneNode}, when two graphs differ by one node $v$, the size difference of their corresponding optimal seed sets for full influenceability can be as large as the number of the neighbors of $v$. Based on this result, Theorem~\ref{THM:ADDV} further shows detailed conditions and expressions to characterize the differences between these two optimal seed sets. Moreover, Theorem~\ref{THM:ADDV} states conditions under which the optimal seed sets for full influenceability can be the same. These observations thus motivate us to remove nodes iteratively with the goal of maintaining the optimal seed set in the remaining graph.
This removal process continues until the objective is accomplished.
Such removal-based method therefore becomes the fundamental principal to our algorithm design for full and partial influenceability. \looseness=-1

\subsection{Seed Selection for Full Influenceability}
\label{subsec:Fullinfluenceability}

\vspace{-.25em}
\begin{algorithm}[tb]
\label{Alg:MSS}
\small
\SetKwInOut{Input}{input}\SetKwInOut{Output}{output}
\Input{Network $\mathcal{G}=(V,L)$, barricade factors $\{b_v\}_{v\in V}$}
\Output{Set of seed nodes $S$}
$U\leftarrow \{u\in \mathcal{V}(\mathcal{G}):b_u\leq \sum_{z\in \mathcal{V}(\mathcal{G})} W_{\overrightarrow{zu}}\}$\tcp*{$\mathcal{V}(\mathcal{G})$: set of nodes in $\mathcal{G}$, $W_{\overrightarrow{zu}}=0$ if $\overrightarrow{zu}$ does not exist\label{mss:identifyU1}}
\While{$|U|>0$\label{mss:removeStart}}
{
\If{$|U|>1$\label{mss:edgeWeightStart}}
{$m_1\leftarrow \min_{u\in U}\big(\sum_{z\in \mathcal{V}(\mathcal{G})} (W_{\overrightarrow{zu}}+W_{\overrightarrow{uz}})\big)$\;
$U\leftarrow \{u\in U: \sum_{z\in \mathcal{V}(\mathcal{G})} (W_{\overrightarrow{zu}}+W_{\overrightarrow{uz}})=m_1\}$\;
}\label{mss:edgeWeightEnd}
\If{$|U|>1$\label{mss:nonTrivialCStart}}
{$m_2\leftarrow \min_{u\in U}$(number of non-trivial connected components in $\mathcal{G}-u$)\tcp*{a connected component containing at least two nodes is \emph{non-trivial}}
$U\leftarrow \{u\in U: \mathcal{G}-u$ has $m_2$ non-trivial connected components$\}$\tcp*{$\mathcal{G}-u$: remove $u$ and all edges (in either direction) incident to $u$ from $\mathcal{G}$}
}\label{mss:nonTrivialCEnd}
\If{$|U|>1$\label{mss:deficientStart}}
{$m_3\leftarrow \min_{u\in U}$(number of influence deficient nodes in $\mathcal{G}-u$)\tcp*{see Definition~\ref{def:deficient}}
$U\leftarrow \{u\in U: \mathcal{G}-u$ has $m_3$ influence deficient nodes$\}$\;
}\label{mss:deficientEnd}
$\mathcal{G}\leftarrow \mathcal{G}-q$\tcp*{$q$: a randomly picked node from $U$\label{mss:removeNode}}
$U\leftarrow \{u\in \mathcal{V}(\mathcal{G}):b_u\leq \sum_{z\in \mathcal{V}(\mathcal{G})} W_{\overrightarrow{zu}}\}$\;\label{mss:identifyU2}
}\label{mss:removeEnd}
$S\leftarrow \mathcal{V}(\mathcal{G})$\tcp*{$\mathcal{V}(\mathcal{G})$: set of nodes in $\mathcal{G}$\label{mss:finalSet}}
\caption{Minimum Seed Selection (MSS)}
\vspace{-.25em}
\end{algorithm}
\normalsize

We note that, depending on the network structure, the conditions in Theorem~\ref{THM:ADDV} are not always easily verifiable, and thus alternative approaches are required. 
In particular, to enable high efficiency of the proposed algorithm, at each iteration in the removal process, we delete the node that least violates the optimal-set-maintaining conditions in Theorem~\ref{THM:ADDV}; this process continues until all remaining nodes are influence deficient nodes (see Definition~\ref{def:deficient}). Based on this rule, we develop Algorithm~\ref{Alg:MSS}, Minimum Seed Selection (MSS), for full influenceability in a given network.

In MSS, lines~\ref{mss:removeStart}--\ref{mss:removeEnd} iteratively remove nodes in the given network. For each removed node, MSS guarantees that it is influenceable if all nodes in the remaining graph are active, while also aiming to minimize its impact on the optimal seed set in the remaining graph. Specifically, at each iteration, MSS first identifies set $U$ containing all non-influence-deficient nodes in line~\ref{mss:identifyU1} (or line~\ref{mss:identifyU2}). For any $u\in U$, $u$ does not need to be a seed node as $u$ can be influenced by the rest of the network (if active); therefore, $u$ can be excluded from the seed set for full influenceability. However, generally $|U|>1$, i.e., there exist ties. According to Theorem~\ref{THM:ADDV}, nodes that do not associate with large edge weights from/to their neighbors may maintain the optimal seed set in the rest of the network. Therefore, lines~\ref{mss:edgeWeightStart}--\ref{mss:edgeWeightEnd} refine $U$ by only keeping the nodes with the minimum impact from/to their neighbors. After this step, if the tie still exists, then lines~\ref{mss:nonTrivialCStart}--\ref{mss:nonTrivialCEnd} further select nodes that incur the minimum number of non-trivial components. Here \emph{non-trivial connected component} after the removal operation refers to a connected component containing at least two nodes, and $\mathcal{G}-u$ denotes removing node $u$ and all edges incident to $u$ from $\mathcal{G}$. Intuitively, when removing a node generates many non-trivial components, a large portion of the original social connections are damaged, thus requiring more seed nodes for full influenceability.
Next, if the tie persists, then lines~\ref{mss:deficientStart}--\ref{mss:deficientEnd} only keep nodes resulting in the minimum number of influence deficient nodes. Finally, a randomly picked node from $U$ (if $|U|>1$) is removed from $\mathcal{G}$ in line~\ref{mss:removeNode}. This removal process is repeated until all nodes in the remaining graph are influence deficient, i.e., $U=\emptyset$; then all remaining nodes form the selected seed set for full influenceability (line~\ref{mss:finalSet}). Note that the output of MSS may not be unique as the node removal sequence is not; the performance of MSS is evaluated in Section~\ref{sect:evaluations}.\looseness=-1

\emph{Correctness of MSS:} MSS can be interpreted as follows. MSS starts by assuming that every node in the given network is a seed node. Then at each iteration, a less-important node that can positively be influenced by the rest of the network is excluded from the seed set; to withdraw this node from consideration for a seed, MSS removes it from the graph and only focuses on seed selection in the rest of the graph. Therefore, at each iteration, we guarantee that if all nodes in the remaining graph are seed nodes, then they can jointly influence all removed nodes in the original graph. Furthermore, let $s$ be the time sequence of how nodes are removed. Then the reverse order of $s$ corresponds to one influence diffusion process starting from the selected seeds. Hence, the seed set selected by MSS can achieve full influenceability.

\emph{Example:} For the sample network in Fig.~\ref{fig:Illustrative_Example}, one possible node removal sequence (recall this sequence is not unique) by MSS is: (i) nodes $v_{19}$, $v_{15}$, $v_{11}$, $v_{16}$, $v_{7}$, $v_{12}$, $v_{17}$, $v_{8}$, $v_{13}$, and $v_{9}$ by the tie breaking rule in lines~\ref{mss:edgeWeightStart}--\ref{mss:edgeWeightEnd}, and (ii) $v_{2}$, $v_{4}$, $v_{6}$, and $v_{14}$ by the tie breaking rule in lines~\ref{mss:nonTrivialCStart}--\ref{mss:deficientEnd}. Then all remaining nodes $S_\mathcal{G}=\{v_{1}, v_{3}, v_{5}, v_{10}, v_{18}\}$ form a seed set for full influenceability. It can be verified that $S_\mathcal{G}$ is optimal as no sets with less number of nodes achieve full influenceability.\looseness=-1

\emph{Complexity:} In MSS, each iteration for removing a non-influence-deficient node takes $O(|U|)$ time as the properties (the impact from/to neighbors in lines~\ref{mss:edgeWeightStart}--\ref{mss:edgeWeightEnd}, the resulting number of non-trivial connected components in lines~\ref{mss:nonTrivialCStart}--\ref{mss:nonTrivialCEnd}, and the resulting number of influence deficient nodes in lines~\ref{mss:deficientStart}--\ref{mss:deficientEnd}) associated with each $u\in U$ can be computed in a constant time. In addition, when updating $U$ after the removal of node $q$ (line~\ref{mss:identifyU2}), only nodes that are neighbors of $q$ (before $q$ is removed) are candidates to be included in $U$, because other nodes have the unchanged sets of neighbors. Therefore, the total time complexity of MSS is $O(cn)=O(n)$ ($n=|V|$ and $c$ is a constant), i.e., linear time complexity.

\vspace{-.25em}
\begin{algorithm}[tb]
\label{Alg:SIM}
\small
\SetKwInOut{Input}{input}\SetKwInOut{Output}{output}
\Input{Network $\mathcal{G}=(V,L)$, barricade factors $\{b_v\}_{v\in V}$, budget $k$}
\Output{Set of seed nodes $S$}
$S\leftarrow$ seed set selected by MSS (Algorithm~\ref{Alg:MSS})\;
\While{$k<|S|$}
{$u=\arg\max_{s\in S} \sigma(S\setminus\{s\})$\tcp*{$\sigma(\cdot)$ is determined by the discrete cascade process as illustrated in Fig.~\ref{fig:Illustrative_Example}\label{sim:minimumImpact}}
$S\leftarrow S\setminus\{u\}$\;\label{sim:removeNode}
}
\caption{Seed selection for Influence Maximization (SIM)}
\vspace{-.25em}
\end{algorithm}
\normalsize

\subsection{Seed Selection for Partial Influenceability}

For full influenceability, Algorithm~\ref{Alg:MSS} removes less-impor\-tant nodes according to the conditions in Theorem~\ref{THM:ADDV} and selects all remaining nodes as seed nodes. Therefore, Algorithm~\ref{Alg:MSS} serves as a tool to identify the critical nodes for influence diffusion. While for partial influenceability with the maximum number of seeds being $k$, we propose Seed selection for Influence Maximization (SIM) to further refine the seed set $S$ generated by MSS by selecting the top-$k$ critical nodes from $S$, as presented in Algorithm~\ref{Alg:SIM}.

In Algorithm~\ref{Alg:SIM}, SIM removes nodes from set $S$ constructed by Algorithm~\ref{Alg:MSS} iteratively until the seed budget $k$ is reached. In particular, at each iteration, the node with the minimum impact (line~\ref{sim:minimumImpact}) on the objective function $\sigma(\cdot)$ is removed from $S$ (line~\ref{sim:removeNode}). When $|S|=k$, SIM returns $S$ as the selected seed set for partial influenceability.

\emph{Complexity:} Since the discrete cascade process can be determined in $O(n)$ time under a given seed set, line~\ref{sim:minimumImpact} in Algorithm~\ref{Alg:SIM} takes $O(n|S|)$ time. Let $m$ denote the size of seed set generated by Algorithm~\ref{Alg:MSS}. Then the total time complexity for Algorithm~\ref{Alg:SIM} is $O(nm^2)$.

\emph{Discussions:} For influence maximization, the well-known\footnote{More efficient algorithms in \cite{Leskovec07KDD,Chen09KDD,Chen10KDD,Jiang11AAAI,Borgs14SODA} rely on the assumption of the objective function being submodular, thus not applicable to our problem.} greedy algorithm \cite{Kempe03KDD} starts with an empty set, and repeatedly selects and adds a node (to this set) that incurs the maximum marginal influence gain of $\sigma(\cdot)$. This process proceeds until the objective (for full influenceability) or the constraint (for partial influenceability under limited budget) is reached. Following similar argument in the complexity analysis of SIM, we know that this greedy algorithm is in $O(n^3)$ complexity. Note that essentially SIM also removes nodes iteratively in a greedy manner; nevertheless, such removal process is conducted on a well-refined and critical set identified by MSS. Moreover, as generally $m\ll n$, SIM experiences a significant improvement over the cubic complexity of a pure greedy algorithm. See Section~\ref{sect:evaluations} for evaluation results, where the above pure greedy algorithm is used as a benchmark.\looseness=-1


\section{Experiments}
\label{sect:evaluations}
To evaluate the performance of MSS and SIM, we conduct
a set of experiments on both randomly-generated and real
Facebook and Twitter networks \cite{snapnets}.
For comparison, we use the greedy algorithm discussed at the end of Section~\ref{sect:AlgSeedSetSelection} as a benchmark.
All algorithms are implemented in Matlab R2010a and
performed on a MacBook Pro with 2.8 GHz Intel Core i7 @ 16 GB memory and OS X EI Capitan.

\subsection{Synthetic Networks}

We first consider synthetic topologies generated according to the widely used Random Geometric (RG) \cite{GuptaKumar99}, Erd\"{o}s-R\'{e}nyi \cite{ErdosRenyi60}, and Random Power Law \cite{Chung06book} graph models. Our motivation for performing evaluations on random networks is that they allow comprehensive evaluation without artifacts of specific network deployments; moreover, the selected graph models provide insights on how topological properties affect MSS and SIM. In this experiment, we obtain similar observations under different graph models; we therefore only report results under the RG model due to page limitations.\looseness=-1

Under the RG model \cite{GuptaKumar99}, nodes are first randomly distributed in a unit square, and then each pair of nodes $u$ and $v$ are connected by edges $\small\overrightarrow{uv}\normalsize$ and $\small\overrightarrow{vu}\normalsize$ if their distance is no larger than a threshold $\theta$.
To focus on evaluating MSS and SIM under different network structural properties and barricade factors, we fix the weight of each edge to $1$ under synthetic networks; the impact under various edge weights will be examined in real Facebook/Twitter networks.

\begin{figure*}[htbp]
  \begin{center}
  \vspace{-.5em}
    \mbox{
      \subfigure[]{\includegraphics[width=0.71\columnwidth]{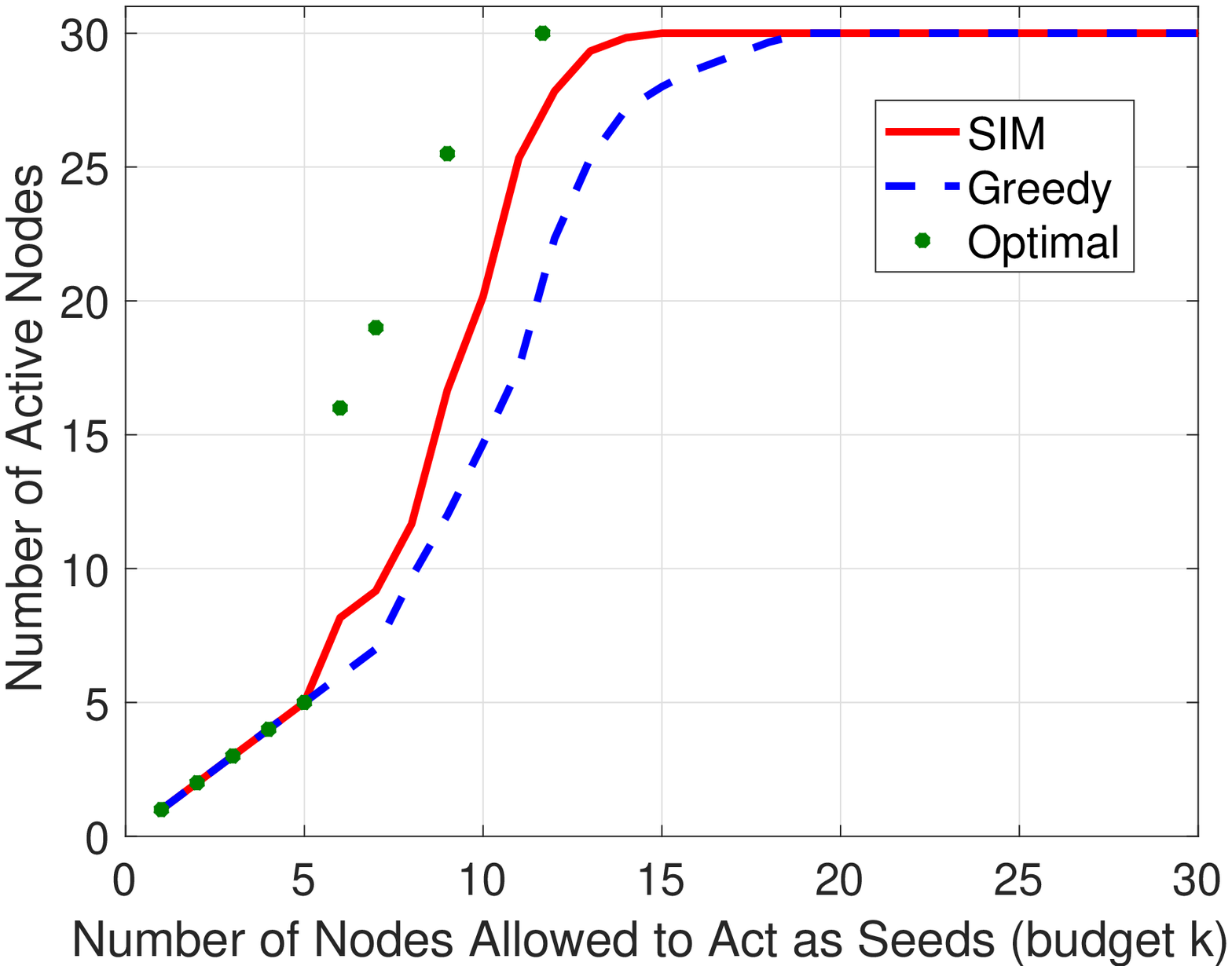}\label{fig:RG30}}
      \subfigure[]{\includegraphics[width=0.71\columnwidth]{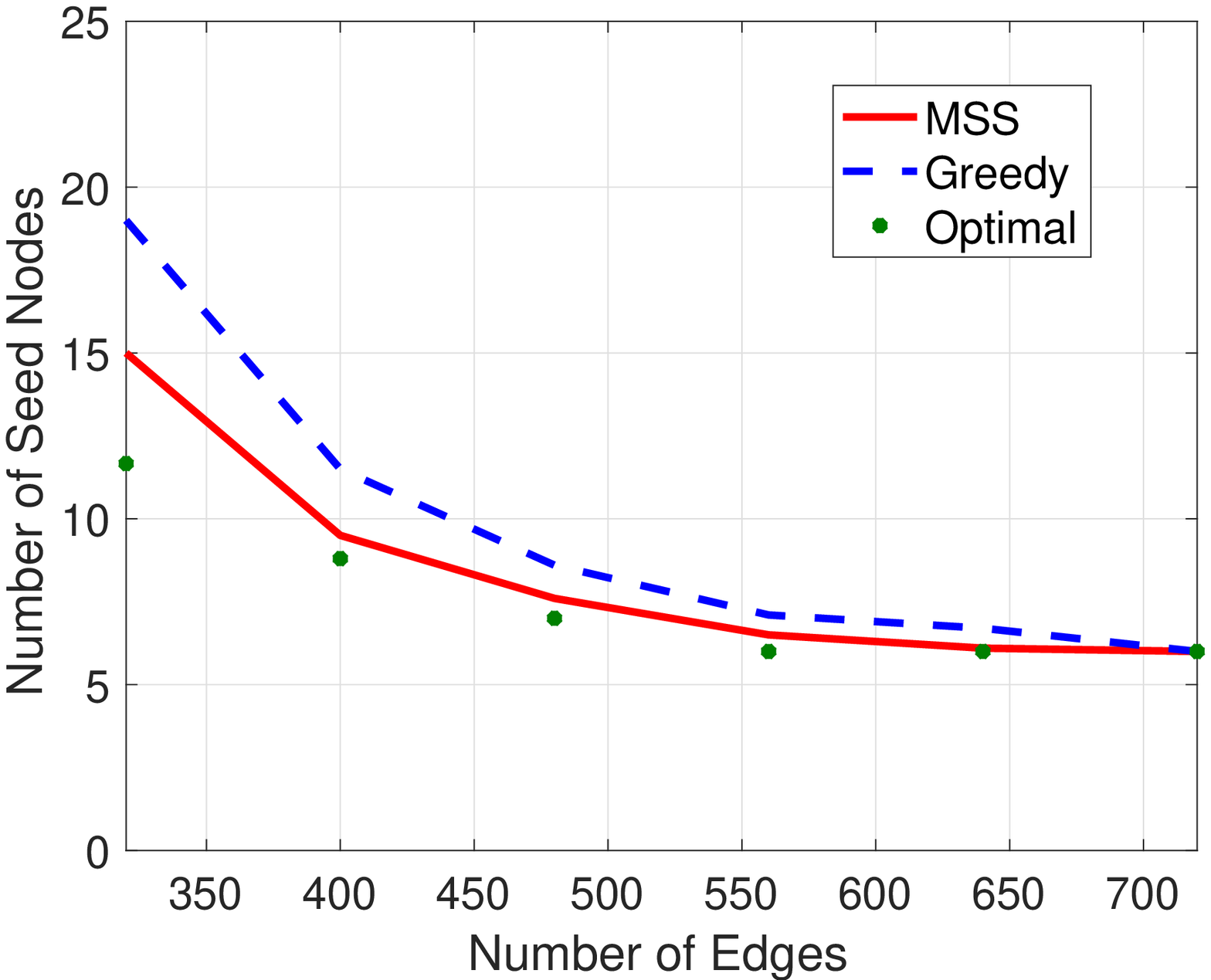}\label{fig:RG30VariableEdges}}
      \subfigure[]{\includegraphics[width=0.71\columnwidth]{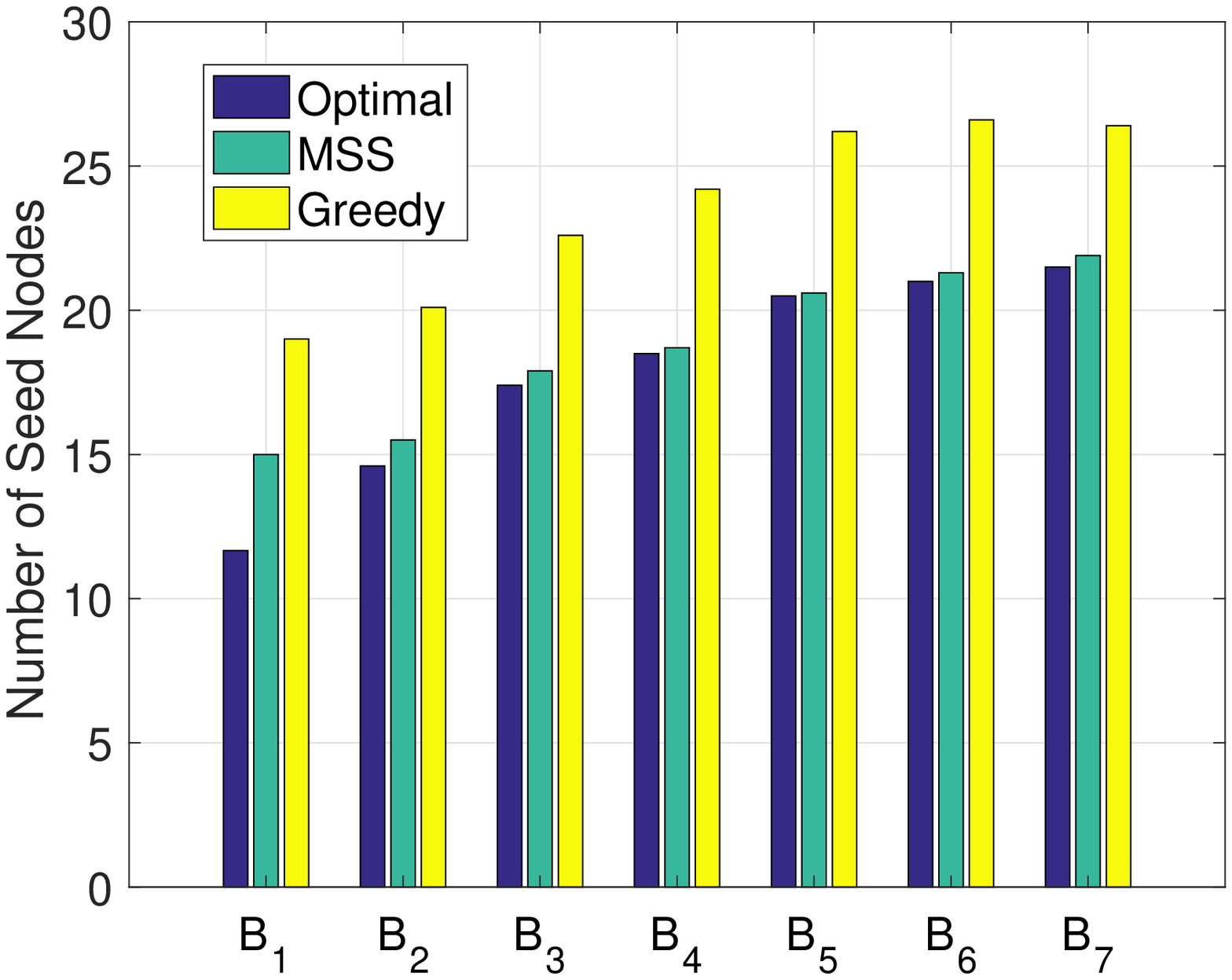}\label{fig:RG30VariableR}}
      }
      \vspace{-1em}
 \caption{Evaluations on Radom Geometric graphs ($|V|=30$, $\forall \protect\overrightarrow{uv}\in L\ W_{\protect\overrightarrow{uv}}=1$, $10$ graph realizations under each parameter setting). (a) Number of achieved active nodes under different budgets: $E[|L|]=320$, $\forall u\in V\ b_u\in [5.33, 10.66]$; (b) Size of the seed set for full influenceability under different $E[|L|]$: $E[|L|]=\{320,400,480,\ldots,720\}$, $\forall u\in V\ b_u\in [5.33, 10.66]$; (c) Size of the seed set for full influenceability under different $\{b_u\}_{u\in V}$: $E[|L|]=320$, $\forall u\in V\ b_u\in [5.33,5.33(0.5i+1.5)]\ $ in $B_i$.}
\label{fig:RG}
  \end{center}
\vspace{-2em}
\end{figure*}


In addition to the greedy algorithm, we also compare our algorithms to the optimal results obtained by enumerating all possible seed sets.
Since it is expensive to compute the optimal results, we randomly generate small RG graph realizations, with each realization containing only $30$ nodes.
The evaluation results averaged over $10$ graph realizations are reported in Fig.~\ref{fig:RG}. In Fig.~\ref{fig:RG}, we first study the selected seed sets under various budgets ($k$). We tune parameter $\theta$ in the RG model to make sure the expected number of (directed) edges is $E[|L|]=320$; furthermore, the barricade factor $b_u$ for each node $u\in V$ is selected from $[{E[|L]}/{(2|V|)},{E[|L]}/{|V|}]$ (i.e., $[5.33,10.66]$) uniformly at random. The corresponding results are shown in Fig.~\ref{fig:RG30}. Fig.~\ref{fig:RG30} illustrates that when the budget $k$ is small, both SIM and the greedy algorithm achieve the optimal value; moreover, $\sigma(S)=k$ when $k\leq 5$. This is because when the selection budget is extremely small, the social influence can hardly propagate, i.e., barricade factors are large. Nevertheless, when $k$ increases, SIM significantly outperforms the greedy benchmark, especially in the case of $k=11$, where SIM almost exhibits a 2-fold improvement. In addition, the optimal result further confirms that the objective function under the barricade model is not submodular. In Fig.~\ref{fig:RG30}, we also note that for a range of $k$ (e.g., $6\leq k\leq 10$),  SIM is unable to accurately approximate the optimal value. Regarding this observation, we argue that SIM is an algorithm that is capable of achieving superior performance over the best-known heuristic (pure greedy solution) while experiencing much smaller complexity. Next, we evaluate the algorithm performance under different network densities (i.e., number of edges to nodes ratio), focusing on seed selection for full influenceability. We adopt the same network parameter settings as those in Fig.~\ref{fig:RG30}, except that $E[|L|]$ is changed to a set of values, i.e., $E[|L|]=\{320,400,480,\ldots,720\}$; the average results are presented in Fig.~\ref{fig:RG30VariableEdges}. In Fig.~\ref{fig:RG30VariableEdges}, it shows that the performance of both MSS and the greedy benchmark improves in networks with high densities, and MSS can even achieves optimality. Furthermore, MSS and the greedy solution converge to the optimal value when the network density is sufficiently high. Intuitively, this is because there exists a large number of optimal seed sets when the network exhibits high density, and thus it is easier for both MSS and the greedy benchmark to find one of them. Finally, we examine how barricade factors may affect the seed selection algorithms. For this goal, let $B_i$ denote $[5.33,5.33(0.5i+1.5)]$ ($i=1,2,\ldots,6$), and barricade factor $b_u$ is selected from $B_i$ uniformly at random for each $i$. We compare the cardinality of the generated seed sets for full influenceability under various $B_i$ ($i=1,2,\ldots,6$), as reported in Fig.~\ref{fig:RG30VariableR}. Fig.~\ref{fig:RG30VariableR} shows that when the average barricade factors are large, MSS accurately approximates the optimal value, owing to the fact that MSS strategically removes nodes unnecessary to serve as seeds. However, the greedy benchmark generates seed sets that are almost $25\%$ more than necessary, thus resulting in a waste of resources.\looseness=-1

\subsection{Facebook and Twitter Networks}

We next use publicly available social network datasets collected by the Stanford SNAP project \cite{snapnets}, from which we select two representative networks, Facebook and Twitter, for algorithm evaluations. For these two selected networks, (i) Facebook is a bidirected graph as the existence of edge $\small\overrightarrow{uv}\normalsize$ directly implies the existence of edge $\small\overrightarrow{vu}\normalsize$, whereas (ii) Twitter is not a bidirected graph as edge $\small\overrightarrow{uv}\normalsize$ does not necessarily imply the existence of edge $\small\overrightarrow{vu}\normalsize$. In this dataset, there are $4,039$ ($107,614$) nodes and $176,468$ ($13,673,453$) directed edges in the Facebook (Twitter) network; the huge sizes of these networks cause extremely long running time in the algorithm evaluation. As such, we randomly sample a smaller subgraph, i.e., subgraph with $500$ ($600$) nodes and $7,280$ ($9,899$) directed edges, from the original Facebook (Twitter) dataset for evaluations. Such subgraph-based influence maximization task also widely exists in real networks. For instance, regarding a product/innovation promotion, the commercial campaigner usually first targets a relatively small geographic region and then identifies the individuals who are most likely to enable vast influence cascades within this region. In this paper, we do not intend to infer edge weights or barricade factors from real datasets; instead, we assume that all these parameters are given and examine all algorithms on top of these network parameters. Since the available data do not include any information about edge weights or barricade factors, we simulate these metrics by randomly generated numbers, and then feed them to the seed selection algorithms.

Using real network data, again we first examine the number of active nodes achieved under various budgets ($k$). In particular, we set edge weights (barricade factors) to be numbers  generated between $1$ and $2$ ($5$ and $10$) uniformly at random. Ten such parameter realizations are generated as the algorithm input; the results averaged over these tests are reported in Fig.~\ref{fig:RealNetworks}. Similar to random graphs, again we observe that SIM outperforms the greedy benchmark by up to $2$-fold (e.g., when $k\approx 25$ in Fig.~\ref{fig:Facebook} and $k\approx 40$ in Fig.~\ref{fig:Twitter}) in both Facebook and Twitter networks. In addition, Fig.~\ref{fig:RealNetworks} demonstrates that the curves generated by SIM are much smoother than those by the benchmark, which partly suggests that our seed selection algorithms are more robust against network parameter changes. Besides, we also report the algorithm average running time for the whole spectrum of $k$, i.e., $k=1,2,\ldots,|V|$, in Fig.~\ref{fig:RealNetworks}. These results show that SIM achieves a roughly $4$-fold speedup, thus confirming the efficiency of SIM.

\begin{figure}[tb]
  \begin{center}
  \vspace{-.5em}
      \subfigure[Facebook ($|V|=500, |L|=7280$
      )]{\includegraphics[width=0.75\columnwidth]{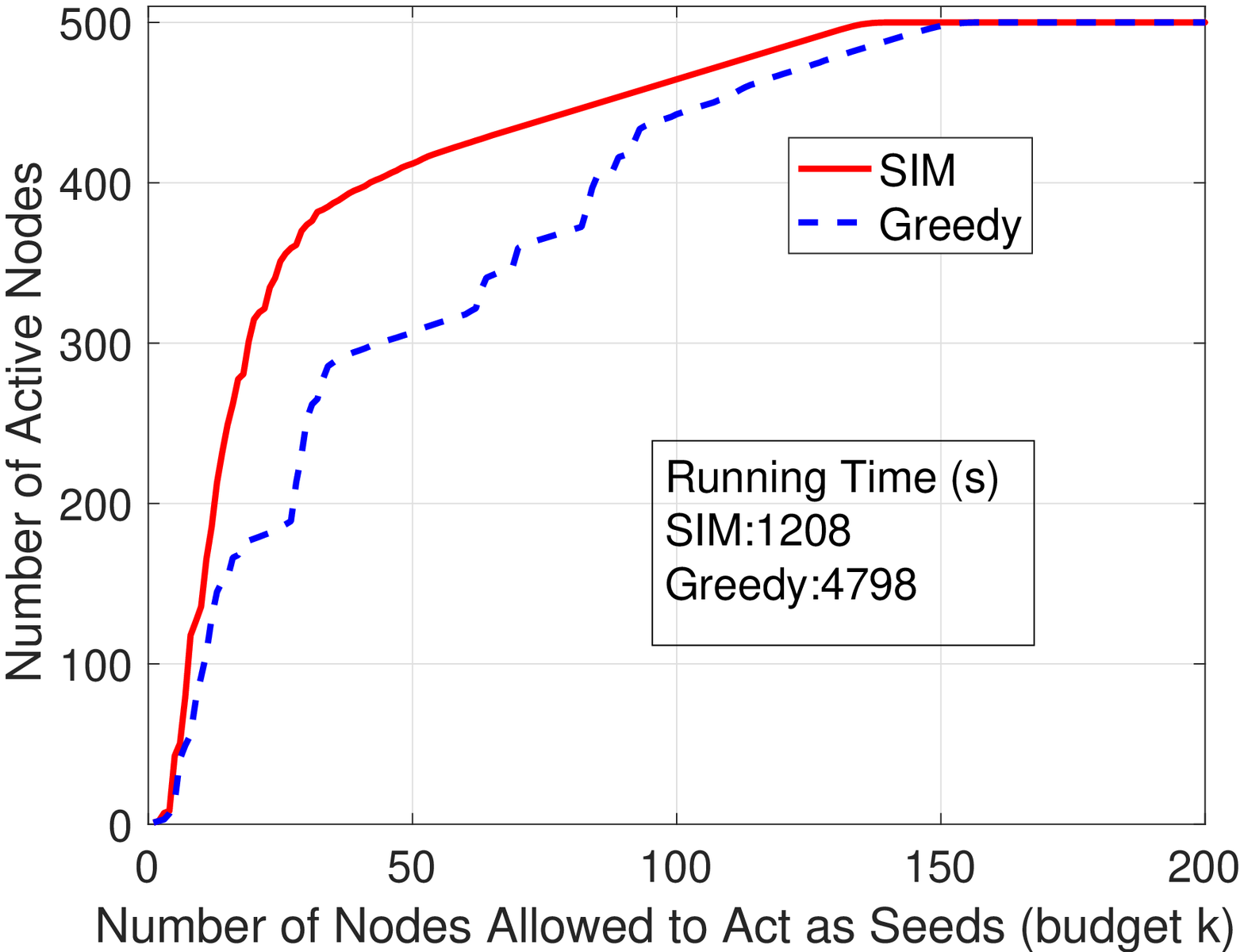}      \label{fig:Facebook}}
      \subfigure[Twitter ($|V|=600, |L|=9899$
      )]{\includegraphics[width=0.75\columnwidth]{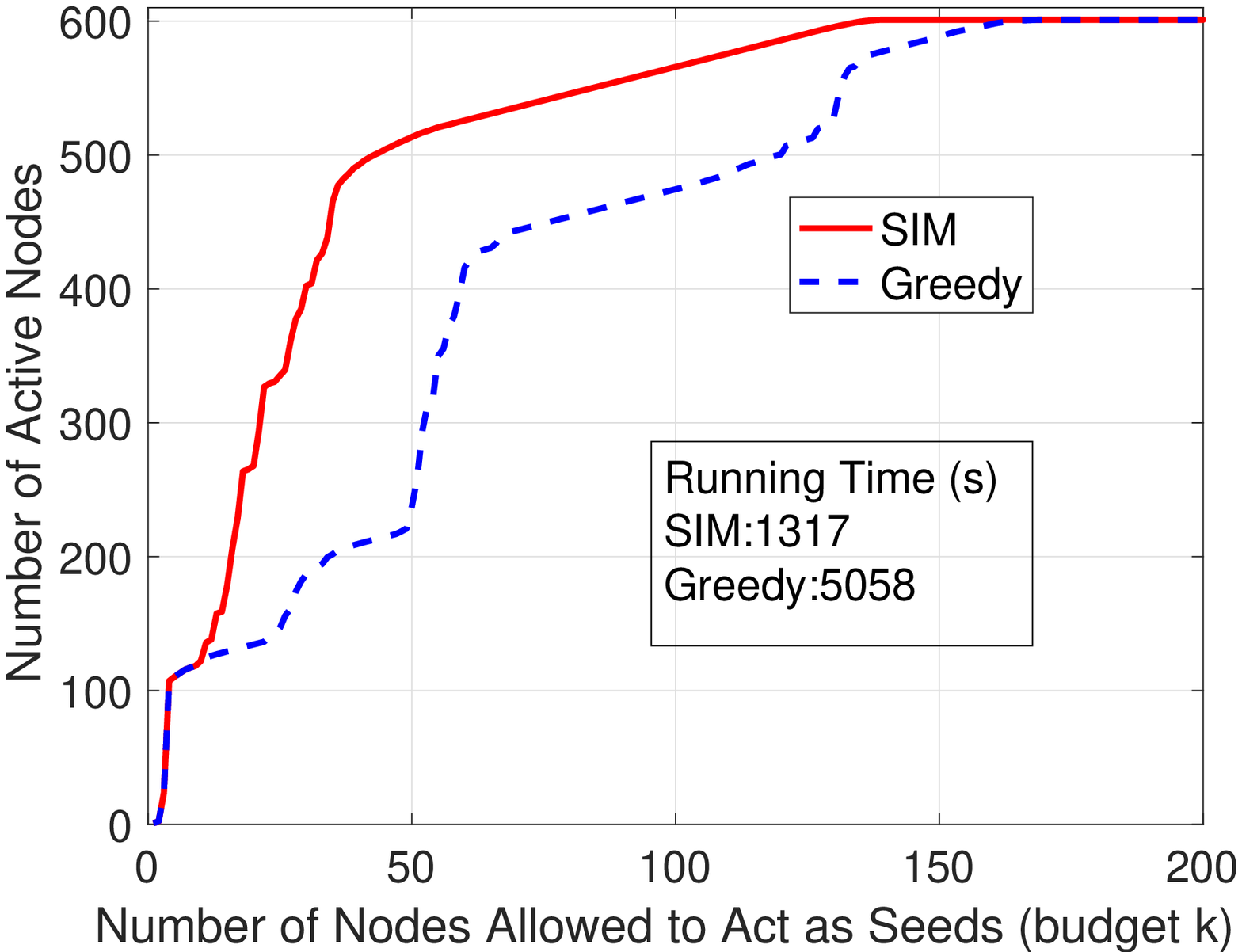}\label{fig:Twitter}}
      \vspace{-1em}
 \caption{Number of achieved active nodes under different budgets in real networks.}
 \vspace{-2em}
\label{fig:RealNetworks}
  \end{center}
\end{figure}

\begin{table}[tb]
\vspace{-.5em}
  \centering
  \caption{Seed Selection in the Facebook/Twitter Network for Full influenceability Under Various Barricade Factors ($\forall u\in V\ b_u\in [5,5(i+1)]\ $ in $B_i$)}\label{t:RealNetworks}
\vspace{-1em}
  \subtable[Facebook Network]{
\begin{tabular}{c|c|c|c|c|c}
\hline
	&	Algorithm	&	$B_1$	&	$B_2$	&	$B_3$	&	$B_4$	\\
\hline											
\#Seed	&	MSS	&	140.0	&	203.0	&	242.8	&	270.8	\\
\cline{2-6}											
Nodes	&	Greedy	&	157.0	&	206.8	&	243.5	&	278.5	\\
\hline											
Running	&	MSS	&	0.9	&	0.6	&	0.5	&	0.4	\\
\cline{2-6}												
Time (s)	&	Greedy	&	4798	&	4597	&	4543	&	4700	\\
\hline																			\end{tabular}
  }
  \subtable[Twitter Network]{
\begin{tabular}{c|c|c|c|c|c}
\hline											
	&	Algorithm	&	$B_1$	&	$B_2$	&	$B_3$	&	$B_4$	\\
\hline											
\#Seed	&	MSS	&	139.0	&	211.0	&	257.0	&	285.3	\\
\cline{2-6}												
Nodes	&	Greedy	&	164.8	&	230.0	&	280.0	&	314.0	\\
\hline											
Running	&	MSS	&	1.3	&	1.0	&	0.9	&	0.8	\\
\cline{2-6}												
Time (s)	&	Greedy	&	5058	&	5180	&	5871	&	5327	\\
\hline																			\end{tabular}
  }
  \vspace{-1em}
\end{table}

Next, we study the algorithm performance under different barricade factors. Similar to Fig.~\ref{fig:RG30VariableR}, we select a set of barricade factors, i.e., $B_i=[5,5(i+1)]$ ($i \in \{1,\ldots,4\}$). Under each $B_i$, $b_u$ of node $u\in V$ has the value selected from $B_i$ uniformly at random, and edge weights are still randomly and uniformly generated between $1$ and $2$. The experiment is repeated $10$ times for each $B_i$; the average results for full influenceability are reported in Table~\ref{t:RealNetworks}. In Table~\ref{t:RealNetworks}, we first observe that MSS outperforms the greedy benchmark in term of the number of selected seed nodes; however, the improvement is not always significant. This is because, as suggested by Fig~\ref{fig:RG30VariableEdges}, when the network density ($E[|L|]/E[|V|]$ is $14.6$ for Facebook and $16.5$ for Twitter) is high, the greedy benchmark is likely to experience improved performance. Besides, we also report the algorithm running time in Table~\ref{t:RealNetworks}, which justifies the momentous advantages of MSS in comparison to the benchmark. In particular, Table~\ref{t:RealNetworks} shows that MSS achieves roughly three orders of magnitude improvement over all ranges of barricade factors for both the Facebook and Twitter networks, which thus confirms the superior efficiency of MSS in finding the seed set for full influenceability.\looseness=-1

In addition, we also perform the same experiment using other randomly sampled subgraphs from the original Facebook and Twitter datasets, and observe similar results. All these evaluation results together certify the efficiency and reliability of our proposed algorithms.

\section{Conclusion}
\label{Sect:Conclusion}

We investigated efficient solutions specifically designed for influence maximization problems under a generic threshold-based non-submodular diffusion model, i.e., influece barricade model. For this issue, we first established a series of theoretical results to relate the optimal seed set to network graphical properties. Employing these results as a basis, we then developed seed selection algorithms to handle non-submodularity in the influence maximization problem by iteratively removing less-critical nodes. Evaluations and comparisons in both synthetic networks and real datasets confirm the efficacy of the proposed graph-based algorithms in identifying a superior set of the most influential nodes as well as the significance in reducing the algorithm execution time.


\appendix
\begin{claim}[Influence Propagation Rule]
\label{clm:InfluencePropRule}
Given the set of active nodes $A_t(S)$ at time $t$ under seed set $S$, nodes can become active at time $t+1$ only if they have at least one in-neighbor in $A_t(S)$, as the influence propagation cannot proceed otherwise.
\end{claim}

\emph{Proof of Theorem~\ref{THM:ADDV}:}
Case (1). When $\sum_{w\in N_{\text{in}}} W_{\overrightarrow{wv}}\geq b_v$, $v$ can be influenced by its neighbors in $N_\text{in}$, i.e., $S^*_{\mathcal{G}}$ also achieves full influenceability in $\mathcal{G}'$, and thus $|S^*_{\mathcal{G}'}|\leq |S^*_{\mathcal{G}}|$. Therefore, we consider if $\exists$ set $S^*_{\mathcal{G}'}$ with $|S^*_{\mathcal{G}'}|<|S^*_{\mathcal{G}}|$. Suppose $\exists S$ with $|S|<|S^*_{\mathcal{G}}|$ and $v\notin S$ that achieves full influenceability in $\mathcal{G}'$. Then $v$ becoming active requires the influence from every node in $N_\text{in}$ as $\forall N'_{\text{in}}\subset N_{\text{in}}$, \small$\sum_{w\in N'_{\text{in}}} W_{\overrightarrow{wv}}<b_v\leq \sum_{w\in N_{\text{in}}} W_{\overrightarrow{wv}}$\normalsize; however, when $v$ is active, it can only influence some nodes in $N_\text{in}$ (because $N_\text{out}\subseteq N_\text{in}$) which is not necessary as nodes in $N_\text{in}$ are already active before $v$ is active. Hence, set $S$ can still achieve full influenceability in graph $\mathcal{G}'-v$ (i.e., $\mathcal{G}$), which contradicts the fact that $S^*_{\mathcal{G}}$ ($|S^*_{\mathcal{G}}|>|S|$) is the minimum seed set for full influenceability in $\mathcal{G}$. Therefore, if $\exists S^*_{\mathcal{G}'}$ with $|S^*_{\mathcal{G}'}|<|S^*_{\mathcal{G}}|$, then we must have $v\in S^*_{\mathcal{G}'}$.
Suppose $\exists S^*_{\mathcal{G}'}$ with $|S^*_{\mathcal{G}'}|<|S^*_{\mathcal{G}}|$. Then selecting $|S^*_{\mathcal{G}}|-2$ nodes from $\mathcal{G}$ along with $v$ as seeds can achieve full influenceability in $\mathcal{G}'$. In graph $\mathcal{G}$ (before node $v$ is added), when $|S^*_{\mathcal{G}}|-2$ seeds are selected, there exist at least two inactive nodes.
Hence, graph $\mathcal{G}$ with $|S^*_{\mathcal{G}}|-2$ seeds can be viewed as a combination of subgraph $\mathcal{G}_{\text{active}}$ containing all active nodes and subgraph $\mathcal{G}_{\text{inactive}}$ containing other nodes.
When a seed node $v$ (we know $v$ must be a seed node in $S^*_{\mathcal{G}'}$) is added and connected to some nodes in $\mathcal{G}$ ($\mathcal{G}$ already contains $|S^*_{\mathcal{G}}|-2$ seeds), full influenceability of $\mathcal{G}'$ can be achieved; in other words, two seeds originally in $\mathcal{G}_{\text{inactive}}$ for full influenceability of $\mathcal{G}$, denoted by $x$ and $y$, can now be replaced by a single seed node $v$ in $\mathcal{G}'$. To make sure this can happen, based on the current active nodes in $\mathcal{G}_{\text{active}}$ and $v$, influence must propagate to $x$ and $y$ in the next time slot, i.e., $\exists S^*_{\mathcal{G}}$ with $\{x,y\}\subset S^*_{\mathcal{G}}$, according to the propagation rule (Claim~\ref{clm:InfluencePropRule}). For such influence propagation, we have $\{x,y\}\subseteq U$, $U=\{u\in\mathcal{V}(\mathcal{G}_{\text{inactive}}):$ $u$ has in-neighbors in $\mathcal{G}_{\text{active}}\}$, as $\forall w\in N_\text{out}$ $W_{\overrightarrow{vw}}<b_w$; in addition, \small$\sum_{w\in \mathcal{V}(\mathcal{G}_\text{active})}W_{\overrightarrow{wi}}< b_i\leq W_{\overrightarrow{vi}}+\sum_{w\in \mathcal{V}(\mathcal{G}_\text{active})}W_{\overrightarrow{wi}}$\normalsize\ ($i=x,y$)
must hold as well. Finally, when $\{x,y\}\subset S^*_{\mathcal{G}}$, it is obvious that edges between $\{x,y\}$ and $\mathcal{G}_{\text{active}}$ are redundant. In sum, $x$ and $y$ satisfy conditions (a--c) in Theorem~\ref{THM:ADDV}. Therefore, $|S^*_{\mathcal{G}'}|<|S^*_{\mathcal{G}}|$ only if $v$ connects to nodes satisfying conditions (a--c), i.e., set $Q$ in Theorem~\ref{THM:ADDV} is not empty. Next, when $|Q|=1$, $v$ can only replace one seed node in $\mathcal{G}$, and thus the size of $S^*_{\mathcal{G}'}$ does not change. Hence, $|S^*_{\mathcal{G}'}|<|S^*_{\mathcal{G}}|$ only if $|Q|>1$. For the sufficiency part, it is obvious that, when $|Q|>1$, all nodes in $Q$ acting as seeds in $S^*_{\mathcal{G}}$ can be replaced by $v$ in $\mathcal{G}'$, i.e., $S^*_{\mathcal{G}'}=(S^*_{\mathcal{G}}\setminus Q)\cup \{v\}$. However, when $|Q|\leq 1$, $S^*_{\mathcal{G}'}=S^*_{\mathcal{G}}$ as $S^*_{\mathcal{G}}$ always achieves full influenceability in $\mathcal{G}'$.

Case (2). When $\sum_{w\in N_{\text{in}}} W_{\overrightarrow{wv}}<b_v$, we have $v\in S^*_{\mathcal{G}'}$ for full influenceability in $\mathcal{G}'$. Then since $\forall w\in N_\text{out}$ $W_{\overrightarrow{vw}}<b_w$, we can follow the same argument in Case~(1) to conclude that $|S^*_{\mathcal{G}'}|<|S^*_{\mathcal{G}}|$ only if $v$ connects to nodes satisfying conditions (a--c), i.e., set $Q$ in Theorem~\ref{THM:ADDV} is not empty. Since $v$ itself must be a seed node, we have $S^*_{\mathcal{G}'}=(S^*_{\mathcal{G}}\setminus Q)\cup \{v\}$ for any $Q$; however, $|S^*_{\mathcal{G}'}|<|S^*_{\mathcal{G}}|$ if and only if $|Q|>1$. \hfill$\blacksquare$

\bibliographystyle{IEEEtran}
\bibliography{mybibSimplifiedA}

\begin{thebibliography}{10}
\providecommand{\url}[1]{#1}
\csname url@samestyle\endcsname
\providecommand{\newblock}{\relax}
\providecommand{\bibinfo}[2]{#2}
\providecommand{\BIBentrySTDinterwordspacing}{\spaceskip=0pt\relax}
\providecommand{\BIBentryALTinterwordstretchfactor}{4}
\providecommand{\BIBentryALTinterwordspacing}{\spaceskip=\fontdimen2\font plus
\BIBentryALTinterwordstretchfactor\fontdimen3\font minus
  \fontdimen4\font\relax}
\providecommand{\BIBforeignlanguage}[2]{{%
\expandafter\ifx\csname l@#1\endcsname\relax
\typeout{** WARNING: IEEEtran.bst: No hyphenation pattern has been}%
\typeout{** loaded for the language `#1'. Using the pattern for}%
\typeout{** the default language instead.}%
\else
\language=\csname l@#1\endcsname
\fi
#2}}
\providecommand{\BIBdecl}{\relax}
\BIBdecl

\bibitem{Domingos01KDD}
P.~Domingos and M.~Richardson, ``Mining the network value of customers,'' in
  \emph{ACM KDD}, 2001.

\bibitem{Richardson02KDD}
M.~Richardson and P.~Domingos, ``Mining knowledge- sharing sites for viral
  marketing,'' in \emph{ACM KDD}, 2002.

\bibitem{Kempe03KDD}
D.~Kempe, J.~Kleinberg, and E.~Tardos, ``Maximizing the spread of influence
  through a social network,'' in \emph{ACM KDD}, 2003.

\bibitem{Nemhauser78MathProg}
G.~L. Nemhauser, L.~A. Wolsey, and M.~L. Fisher, ``An analysis of
  approximations for maximizing submodular set functions {I},''
  \emph{Mathematical Programming}, vol.~14, no.~1, pp. 265--294, 1978.

\bibitem{Leskovec07KDD}
J.~Leskovec, A.~Krause, C.~Guestrin, C.~Faloutsos, J.~VanBriesen, and
  N.~Glance, ``Cost-effective outbreak detection in networks,'' in \emph{ACM
  KDD}, 2007.

\bibitem{Chen09KDD}
W.~Chen, Y.~Wang, and S.~Yang, ``Efficient influence maximization in social
  networks,'' in \emph{ACM KDD}, 2009.

\bibitem{Chen10KDD}
W.~Chen, C.~Wang, and Y.~Wang, ``Scalable influence maximization for prevalent
  viral marketing in large-scale social networks,'' in \emph{ACM KDD}, 2010.

\bibitem{Jiang11AAAI}
Q.~Jiang, G.~Song, G.~Cong, Y.~Wang, W.~Si, and K.~Xie, ``Simulated annealing
  based influence maximization in social networks,'' in \emph{AAAI}, 2011.

\bibitem{Borgs14SODA}
C.~Borgs, M.~Brautbar, J.~Chayes, and B.~Lucier, ``Maximizing social influence
  in nearly optimal time,'' in \emph{SIAM SODA}, 2014.

\bibitem{Ma08CIKM}
H.~Ma, H.~Yang, M.~R. Lyu, and I.~King, ``Mining social networks using heat
  diffusion processes for marketing candidates selection,'' in \emph{ACM CIKM},
  2008.

\bibitem{Zhang13ICDCS}
H.~Zhang, T.~N. Dinh, and M.~T. Thai, ``Maximizing the spread of positive
  influence in online social networks,'' in \emph{IEEE ICDCS}, 2013.

\bibitem{Zhu16INFOCOM}
Y.~Zhu, D.~Li, and Z.~Zhang, ``Minimum cost seed set for competitive social
  influence,'' in \emph{IEEE INFOCOM}, 2016.

\bibitem{Lu15VLDB}
W.~Lu, W.~Chen, and L.~V.~S. Lakshmanan, ``From competition to complementarity:
  {C}omparative influence diffusion and maximization,'' \emph{Proc. VLDB
  Endow.}, vol.~9, no.~2, pp. 60--71, 2015.

\bibitem{Lin17ICDE}
Y.~Lin, W.~Chen, and J.~C.~S. Lui, ``Boosting information spread: {A}n
  algorithmic approach,'' in \emph{IEEE ICDE}, 2017.

\bibitem{Anshelevich15AAMAS}
E.~Anshelevich, A.~Hate, and M.~Magdon-Ismail, ``Seeding influential nodes in
  non-submodular models of information diffusion,'' \emph{Autonomous Agents and
  Multi-Agent Systems}, vol.~29, no.~1, pp. 131--159, 2015.

\bibitem{Zhang14KDD}
P.~Zhang, W.~Chen, X.~Sun, Y.~Wang, and J.~Zhang, ``Minimizing seed set
  selection with probabilistic coverage guarantee in a social network,'' in
  \emph{ACM KDD}, 2014.

\bibitem{Karampourniotis2017}
\BIBentryALTinterwordspacing
P.~D. Karampourniotis, ``The impact of heterogeneity on threshold-limited
  social contagion, and on crowd decision-making,'' Ph.D. dissertation,
  Rensselaer Polytechnic Institute, 2017. [Online]. Available:
  \url{http://www.cs.rpi.edu/~szymansk/theses/Panos_thesis.17.pdf}
\BIBentrySTDinterwordspacing

\bibitem{Yang12CIKM}
D.-N. Yang, W.-C. Lee, N.-H. Chia, M.~Ye, and H.-J. Hung, ``On bundle
  configuration for viral marketing in social networks,'' in \emph{ACM CIKM},
  2012.

\bibitem{Chen09SDM}
N.~Chen, ``On the approximability of influence in social networks,'' \emph{SIAM
  Journal on Discrete Mathematics}, vol.~23, no.~3, pp. 1400--1415, 2009.

\bibitem{Granovetter78}
M.~Granovetter, ``Threshold models of collective behavior,'' \emph{American
  Journal of Sociology}, vol.~83, no.~6, pp. 1420--1443, 1978.

\bibitem{Aral11Mgmt}
S.~Aral and D.~Walker, ``Creating social contagion through viral product
  design: A randomized trial of peer influence in networks,'' \emph{Management
  Science}, vol.~57, no.~9, pp. 1623--1639, 2011.

\bibitem{Liu12ICDM}
B.~Liu, G.~Cong, D.~Xu, and Y.~Zeng, ``Time constrained influence maximization
  in social networks,'' in \emph{IEEE ICDM}, 2012.

\bibitem{Chen12AAAI}
W.~Chen, W.~Lu, and N.~Zhang, ``Time-critical influence maximization in social
  networks with time-delayed diffusion process,'' in \emph{AAAI}, 2012.

\bibitem{Lu16TMC}
Z.~Lu, Y.~Wen, W.~Zhang, Q.~Zheng, and G.~Cao, ``Towards information diffusion
  in mobile social networks,'' \emph{IEEE Transactions on Mobile Computing},
  vol.~15, no.~5, pp. 1292--1304, 2016.

\bibitem{Carnes07ICEC}
T.~Carnes, C.~Nagarajan, S.~M. Wild, and A.~van Zuylen, ``Maximizing influence
  in a competitive social network: {A} follower's perspective,'' in \emph{ACM
  ICEC}, 2007.

\bibitem{Lin15Performance}
Y.~Lin and J.~C. Lui, ``Analyzing competitive influence maximization problems
  with partial information: {A}n approximation algorithmic framework,''
  \emph{Performance Evaluation}, vol.~91, pp. 187 -- 204, 2015.

\bibitem{Bakshy12WWW}
E.~Bakshy, I.~Rosenn, C.~Marlow, and L.~Adamic, ``The role of social networks
  in information diffusion,'' in \emph{WWW}, 2012.

\bibitem{snapnets}
\BIBentryALTinterwordspacing
J.~Leskovec and A.~Krevl, ``{SNAP Datasets}: {Stanford} large network dataset
  collection,'' June 2014. [Online]. Available:
  \url{http://snap.stanford.edu/data}
\BIBentrySTDinterwordspacing

\bibitem{GuptaKumar99}
P.~Gupta and P.~Kumar, ``Critical power for asymptotic connectivity in wireless
  networks,'' \emph{Stochastic Analysis, Control, Optimization and
  Applications}, pp. 547--566, 1999.

\bibitem{ErdosRenyi60}
P.~Erd\"{o}s and A.~R\'{e}nyi, ``On the evolution of random graphs,''
  \emph{Publications of the Mathematical Institute of the Hungarian Academy of
  Sciences}, vol.~5, pp. 17--61, 1960.

\bibitem{Chung06book}
F.~Chung and L.~Lu, \emph{Complex Graphs and Networks}.\hskip 1em plus 0.5em
  minus 0.4em\relax American Mathematical Society, 2006.

\end{thebibliography}

\end{document}